\documentclass[twocolumn,english,preprintnumbers,amsmath,amssymb]{revtex4}
\usepackage{dsfont, comment}

\begin{document}

\preprint{MIT-CTP 4134}

\title{QMA-complete problems for stoquastic Hamiltonians and Markov matrices}

\author{Stephen P. Jordan}

\affiliation{Institute for Quantum Information, California Institute of Technology,
Pasadena, CA 91125}

\author{David Gosset}

\affiliation{Center for Theoretical Physics, Massachusetts Institute of Technology,
77 Massachusetts Ave, 6-304, Cambridge, MA 02139 }

\author{Peter J. Love}

\affiliation{Department of Physics, Haverford College, 370 Lancaster Avenue, Haverford,
PA 19041}

\affiliation{Institute for Quantum Information, California Institute of Technology,
Pasadena, CA 91125}

\bibliographystyle{plain}

\newcommand{\p}{\ddots}
\newcommand{\ud}{\mathrm{d}}
\newcommand{\bra}[1]{\langle #1|}
\newcommand{\ket}[1]{|#1\rangle}
\newcommand{\braket}[2]{\langle #1|#2\rangle}
\newcommand{\Bra}[1]{\left<#1\right|}
\newcommand{\Ket}[1]{\left|#1\right>}
\newcommand{\Braket}[2]{\left< #1 \right| #2 \right>}
\renewcommand{\th}{^\mathrm{th}}
\newcommand{\tr}{\mathrm{Tr}}
\newcommand{\id}{\mathds{1}}
\newcommand{\ketbra}[2]{| #1 \rangle \langle #2 |}

\newenvironment{proof}{\noindent \textbf{Proof:}}{$\Box$}

\newtheorem{definition}{Definition}
\newtheorem{lemma}{Lemma}
\newtheorem{theorem}{Theorem}
\newtheorem{prop}{Proposition}

\begin{abstract}
We show that finding the lowest eigenvalue of a 3-local symmetric
stochastic matrix is QMA-complete. We also show that finding the highest
energy of a stoquastic Hamiltonian is QMA-complete and that adiabatic
quantum computation using certain excited states of a stoquastic Hamiltonian
is universal. We also show that adiabatic evolution in the ground
state of a stochastic frustration free Hamiltonian is universal. Our
results give a new QMA-complete problem arising in the classical
setting of Markov chains, and new adiabatically universal Hamiltonians
that arise in many physical systems. 
\end{abstract}

\maketitle

\section{Introduction}
\label{sec_intro}

Quantum complexity theory is the study of the capabilities and limitations
of computational devices operating according to the principles of
quantum mechanics~\cite{QCT}. Because many of the classical constructs
of computer science (e.g. circuits and clauses) are replaced by matrices,
quantum complexity theory is sometimes referred to as {\em
matrix-valued} complexity theory~\cite{Bravyi2,Bravyi_Terhal}. In addition
to its intrinsic interest this subject has many connections to issues
of practical relevance to physical science, such as the difficulty
of computing properties of quantum systems using either quantum or
classical devices~\cite{SchuchVer2009,LiuVer2007,Bravyi1}.

Perhaps the most basic classical complexity classes are P - the class
of problems solved by a deterministic Turing machine in polynomial
time, and NP - the class of problems whose verification lies
in P. It is widely believed, but not proven, that NP is strictly
larger than P~\cite{Sipser}.

Because quantum mechanics only predicts probabilities of events, the
classical deterministic classes are not the most natural place to
start if one seeks their quantum generalizations. The probabilistic
generalization of P is BPP (Bounded-error Probabilistic
Polynomial-time) - those problems solvable by a probabilistic Turing machine
in polynomial time with bounded error~\cite{BPP}. The quantum generalization
of this class is BQP (Bounded-error Quantum Polynomial-time) - the
class of problems solvable in polynomial time with bounded error on
a quantum computer~\cite{QCT}.

The classical probabilistic generalization of NP is the class MA~\cite{MAAM}.
This generalizes NP to problems whose verification is in BPP. MA stands
for Merlin-Arthur. Merlin, who is computationally unbounded but
untrustworthy, provides a proof that Arthur can verify using his BPP
machine. The class MA possesses a quantum generalization to QMA
(Quantum Merlin Arthur)~\cite{WatrousQMA,Kitaev_book,QNP}. QMA may be
intuitively understood as the class of decision problems that can be
efficiently verified by a quantum computer.

Given a classical description of a decision problem $x$ of length
$n$, the prover, Merlin, provides a witness state $\ket{\psi}$ to
the verifier, Arthur. Arthur then peforms a $\mathrm{poly}(n)$-time quantum
computation on the witness $\ket{\psi}$ and either accepts or
rejects. A problem is contained in QMA if, for all YES instances,
there exists a witness causing Arthur to accept with probability
greater than $2/3$ and for NO instances, there does not exist any
witness that causes Arthur to accept with probability greater than
$1/3$. A problem $X$ is said to be QMA-complete if it is contained in QMA
and every problem in QMA can be converted to an instance of $X$ in
classical polynomial time.

Let us consider the following question: What is the ground state energy
of a quantum system? This question lies at the core of many areas
of physical science, including electronic structure theory and condensed
matter physics. In quantum complexity theory this problem has been
formalized (originally by Kitaev~\cite{Kitaev_book}, see also, for
example,~\cite{Kempe}) as the $k$-local Hamiltonian problem. For some
systems, complexity-theoretic arguments suggest that efficient
computation of the ground state energy is likely to remain beyond
reach~\cite{barahona,Kitaev_book}.

A Hamiltonian $H$, acting on $n$ qubits, is said to be $k$-local if it
is of the form 
\[
H=\sum_{s}H_{s},
\]
where each $H_{s}$ acts on at most $k$-qubits. Thus, for example,
1-local Hamiltonians consist only of external fields acting on individual
qubits, and 2-local Hamiltonians consist of 1-local terms and pairwise
couplings between qubits. Physically realistic Hamiltonians are usually
$k$-local with small $k$, often $2$ or $1$, and each local term has
bounded norm. Note that this notion of locality has nothing to do with
spatial locality; a 2-local Hamiltonian may have long-range couplings
but they must be pairwise.\\
\\
\textbf{Problem:} $k$-local Hamiltonian \\
 \textbf{Input:} We are given a classical description of a $k$-local
 Hamiltonian $H$ on $n$ qubits $H=\sum_{j=1}^{r}H_{j}$ with
 $r=\mathrm{poly}(n)$. Each $H_{j}$ acts on at most $k$ qubits and has
 $O(1)$ operator norm. In addition we are given two constants $a$ and
 $b$ such that $0\leq a\leq b$, and $b-a=\epsilon>1/\mathrm{poly}(n)$.\\
 \textbf{Output:} If $H$ has an eigenvalue $\leq a$ answer YES.
If all eigenvalues of $H$ are $>b$ answer NO.\\
 \textbf{Promise:} the Hamiltonian is such that it will produce
either YES or NO.\\
\\
Perhaps a more obvious formulation of this problem is to
ask for an approximate ground state energy to within $\pm\epsilon$
of the correct answer. However, if one can decide the answer to $k$-local
Hamiltonian in polynomial time, then one can solve the approximation
version in polynomial time by a binary search. Thus, the approximation
problem is of equivalent difficulty to the {}``decision'' version,
to within a polynomial factor.

The problem $k$-local Hamiltonian is QMA-complete for
$k\geq2$~\cite{Kempe}. The $k$-local Hamiltonian problem is specified
by the matrix elements of the local terms of $H$. YES instances
possess the ground state as a witness. The verification circuit is the
phase estimation algorithm - a suitably formalized version of the
notion of energy measurement~\cite{Kitaev_book}. If the lowest
eigenvalue of $H$ is less than $a$ (a YES instance) then Arthur will
accept the ground state as a witness. However, if the lowest
eigenvalue of $H$ is greater than $b$ (a NO instance) then Merlin
cannot supply any eigenstate or superposition of eigenstates that
will result in a measurement of energy less than $b$. 

It is considered unlikely that $\mathrm{QMA} \subseteq \mathrm{BQP}$
and therefore it is probably impossible to construct a general quantum
(or classical) algorithm that
finds ground state energies in polynomial time. However, many Hamiltonians
studied in practice have additional restrictions beyond
$k$-locality. In particular, many physical systems are
\emph{stoquastic}, meaning that all of their off-diagonal matrix
elements are nonpositive in the standard basis. This includes the
ferromagnetic Heisenberg model, the quantum transverse Ising model,
and most Hamiltonians achievable with Josephson-junction flux
qubits~\cite{Bravyi1}. In~\cite{Bravyi1} it was shown that for any
fixed $k$, stoquastic $k$-local Hamiltonian is contained in the
complexity class AM. Thus, unless $\mathrm{QMA}\subseteq\mathrm{AM}$
(which is believed to be unlikely), stoquastic $k$-local
Hamiltonian is not QMA-complete\footnote{Like MA, the class AM is a
  probabilistic generalization of NP, see \cite{Arora}.}. It was also
shown in~\cite{Bravyi1} that, for any fixed $k$, adiabatic quantum
computation in the ground state of a $k$-local stoquastic Hamiltonian
can be simulated in $\mathrm{BPP}_{\mathrm{path}}$. Thus, unless
$\mathrm{BQP}\subseteq\mathrm{BPP}_{\mathrm{path}}$ (which is also
believed to be unlikely), such quantum computation is not
universal. The work of~\cite{Bravyi2} also defines a random stoquastic
local Hamiltonian problem which is complete for the class AM.

These results were tightened further for stoquastic frustration
free (SFF) Hamiltonians in~\cite{Bravyi_Terhal}. A local Hamiltonian is
frustration free if it can be written as a sum of terms
\begin{equation} 
H=\sum_{s}^{m}H_{s},
\end{equation}
such that 
\begin{enumerate}
\item {Each local operator $H_{s}$ is positive semidefinite} 
\item {The ground state $|\psi\rangle$ of H satisfies $H_{s}|\psi\rangle=0$
for each $s\in\{1,\dots,m\}$} 
\end{enumerate}
The work of ~\cite{Bravyi_Terhal} showed that an adiabatic evolution along
a path composed entirely of SFF Hamiltonians may be simulated by a
sequence of classical random walks - that is, the adiabatic evolution
may be simulated in the complexity class BPP.

These results were extended to the quantum $k$-satisfiability problem
in~\cite{Bravyi2,Bravyi_Terhal}. The quantum $k$-satisfiability problem
was defined in~\cite{Bravyi4} and we reproduce the definition here:\\
\\
\textbf{Problem:} Quantum $k$-SAT \\
 \textbf{Input:} A set of $k$-local projectors $\{\Pi_{q}\}$
for $q\in\{1,\dots,m\}$ where $m=\mathrm{poly}(n)$ and a parameter
 $\tilde{\epsilon}>1/\mathrm{poly}(n)$\\
\textbf{Output:} If there is a state $\ket{\phi}$ such that $\Pi_{1}\ket{\phi}=0$
for each $q\in\{1,\dots,M\}$, then this is a YES instance. If every
state $\ket{\phi}$ satisfies \[
\sum_{q=1}^{M}\bra{\phi}\Pi_{q}\ket{\phi}\geq\tilde{\epsilon}\]
 then it is a NO instance.\\
 \textbf{Promise:} The instance is either YES or NO. \\
\\
In~\cite{Bravyi2} the stoquastic restriction of quantum $k$-SAT was
shown to be contained in MA for any constant $k$, and MA-complete
for $k=6$ -- the first nontrivial example of an MA-complete
problem. In~\cite{Bravyi_Terhal} these results were extended to a
simplified form of stoquastic quantum $k$-SAT in which projectors
$\Pi_{a}$ all have matrix elements taken from the set $\{0,1/2,1\}$,
and the stoquastic constraints which appear as terms in the
Hamiltonian are of the form $H_{a}=\id-\Pi_{a}$.

The main intuition behind these results is that, by the Perron-Frobenius
theorem, the ground state of a stoquastic Hamiltonian consists entirely
of real positive amplitudes (given the appropriate choice of global
phase). Thus the ground state is proportional to a classical probability
distribution. For this reason, ground state properties are amenable
to classical random walk algorithms and certain problems such as
stoquastic $k$-local Hamiltonian fall into classical probabilistic
complexity classes such as AM. Diffusion Quantum Monte Carlo calculations for
stoquastic Hamiltonians do not suffer from the sign problem because
the negativity of the nonzero off-diagonal matrix elements guarantees
that the transition probabilities in the associated random walk are
all positive.

In this paper we first demonstrate that stoquastic
Hamiltonians may be constructed which allow universal adiabatic quantum
computation in a subspace. Then we show that the $3$-local Hamiltonian problem
is QMA-complete when restricted to stochastic Hamiltonians. These are
Hamiltonians in which all matrix elements are real and nonnegative,
and the sum of matrix elements in any row or column is one. Hence
determining the lowest eigenstate of a symmetric stochastic matrix is
QMA-hard. If $H$ is a stochastic Hamiltonian, then $-H$ is
stoquastic. Thus, our result also shows that determination of the
highest lying eigenstate of a stoquastic matrix is QMA-hard,
sharpening the intuition that it is the positivity of the ground state
which causes its local Hamiltonian problem to fall in a classical
class. We then show that universal adiabatic quantum computation is
possible in the ground state of a stochastic frustration free
Hamiltonian. Defining the computational problem stochastic $k$-SAT in
analogy to the definition of stoquastic $k$-SAT given in
\cite{Bravyi_Terhal}, we show that this problem is QMA$_{1}$-complete for
$k=6$. (QMA$_1$ is a slight variant of QMA such that in YES instances,
Arthur can be made to accept with probability one\cite{Bravyi4}.)

\section{QMA-completeness and adiabatic universality of stoquastic Hamiltonians}
\label{QMA_section}

We start with the result of~\cite{Biamonte}, which
shows that for a Hamiltonian of the form 
\begin{equation}
\label{xzform}
H_{XZ} = \sum_{i}d_{i}X_{i} + \sum_{i}h_{i}Z_{i} +
\sum_{i,j}K_{ij}X_{i}X_{j} + \sum_{i,j}J_{ij}Z_{i}Z_{j},
\end{equation}
the 2-local Hamiltonian problem is QMA-complete if the coefficients
$d_{i}$, $h_{i}$, $K_{ij}$, $J_{ij}$ are allowed to have both
signs. Furthermore, time-dependent Hamilonians that take the form
$H_{XZ}$ at all times can perform universal adiabatic quantum
computation~\cite{Biamonte}.

Starting with a Hamiltonian of the form $H_{XZ}$ on $n$ qubits we
can eliminate the negative matrix elements in each term using a technique
from~\cite{Wocjan_walks}. Essentially, the idea is that instead
of representing the group $Z_{2}$ by $\{1,-1\}$ we use its regular
representation: \[
\left\{ \left[\begin{array}{cc}
1 & 0\\
0 & 1\end{array}\right],\left[\begin{array}{cc}
0 & 1\\
1 & 0\end{array}\right]\right\} .\]
 $H_{XZ}$ can be rewritten as \begin{equation}
H_{XZ}=-\sum_{k}\alpha_{k}T_{k}\label{alpha}\end{equation}
 Where each coefficient $\alpha_{k}$ is positive and for each $k$,
$T_{k}$ is one of \begin{equation}
\pm X,\pm Z,\pm X_{i}X_{j},\pm Z_{i}Z_{j}\label{possibilities}\end{equation}
 with identity acting on the remaining qubits. For any $k$, $T_{k}$
is a $2^{n}\times2^{n}$ matrix in which each entry is either +1,-1,
or 0. From $T_{k}$ we construct a $2^{n+1}\times2^{n+1}$ matrix
$\widetilde{T}_{k}$ by making the following replacements \begin{equation}
1\to\left[\begin{array}{cc}
1 & 0\\
0 & 1\end{array}\right],\ -1 \to \left[\begin{array}{cc}
0 & 1\\
1 & 0\end{array}\right],\ 0 \to \left[\begin{array}{cc}
0 & 0\\
0 & 0\end{array}\right].\label{replacement}\end{equation}
 We can interpret $\widetilde{T}_{k}$ as acting on $n+1$ qubits.
The $2\times2$ matrices of (\ref{replacement}) act on the
ancilla qubit that has been added. Each $T_{k}$ is 2-local or 1-local,
thus each corresponding $\widetilde{T}_{k}$ is 3-local or 2-local.
Furthermore, each $\widetilde{T}_{k}$ is a permutation matrix. Let
\begin{equation}
\widetilde{H}_{XZ}=-\sum_{k}\alpha_{k}\widetilde{T}_{k}.\end{equation}
 This is a linear combination of permutation matrices with negative
coefficients. By construction, $\widetilde{H}_{XZ}$ is therefore
a 3-local stoquastic Hamiltonian. We can rewrite $\widetilde{H}_{XZ}$
as 
\begin{equation}
\widetilde{H}_{XZ}=H_{XZ}\otimes \ketbra{-}{-} - \bar{H}_{XZ} \otimes
 \ketbra{+}{+} \label{nicestoq}
\end{equation}
where
\begin{equation}
\bar{H}_{XZ}=\sum_{k}\alpha_{k}|T_{k}|,
\end{equation}
$|T_{k}|$ is the entry-wise absolute value of $T_k$, and
\begin{eqnarray*}
\ket + & = & \frac{1}{\sqrt{2}}\left(\ket 0+\ket 1\right)\\
\ket - & = & \frac{1}{\sqrt{2}}\left(\ket 0-\ket 1\right).
\end{eqnarray*}
The projectors $\Ket -\Bra -$ and $\Ket +\Bra +$ act on the ancilla
qubit.

Equation \ref{nicestoq} makes the relationship between the spectra of
$H_{XZ}$ and $\widetilde{H}_{XZ}$ clear. Let
$\ket{\psi_{0}},\ket{\psi_{1}},\ldots,\ket{\psi_{N-1}}$ denote the
eigenstates of $H_{XZ}$ with corresponding eigenvalues
$\lambda_{0}\leq\lambda_{1}\leq\ldots\leq\lambda_{N-1}$, and let
$\ket{\bar{\psi}_{0}},\ket{\bar{\psi}_{1}},\ldots,\ket{\bar{\psi}_{N-1}}$
denote the eigenstates of $\bar{H}_{XZ}$ with corresponding
eigenvalues
$\bar{\lambda}_{0}\leq\bar{\lambda}_{1}\leq\ldots\leq\bar{\lambda}_{N-1}$.
($H_{XZ}$ acts on $n$ qubits, so $N=2^{n}$.) $\widetilde{H}_{XZ}$,
which acts on a $2N$-dimensional Hilbert space, has two
$N$-dimensional invariant subspaces. The first is spanned by
$\Ket{\psi_{j}}\Ket -$ with eigenvalues $\lambda_{j}$. The second is
spanned by eigenvectors $\Ket{\bar{\psi}_{j}}\Ket +$ with eigenvalues
$-\bar{\lambda}_{j}$.

We can perform universal adiabatic quantum computation in such an
eigenstate of a stoquastic Hamiltonian. To prove this, we make use of
the universal adiabatic Hamiltonian $H_{XZ}(t)$ from \cite{Biamonte},
which at all $t$ takes the form shown in equation \ref{xzform}.
One can use the construction described above to obtain a
stoquastic Hamiltonian $\widetilde{H}_{XZ}(t)$ corresponding to each
instantaneous Hamiltonian $H_{XZ}(t)$. In this way we obtain a time
varying Hamiltonian $\widetilde{H}_{XZ}(t)$ whose spectrum in the
$\ket -$ subspace exactly matches the spectrum of $H_{XZ}(t)$, the
only difference being the addition of an ancilla qubit in the $\ket -$
state. Because $\widetilde{H}_{XZ}(t)$ has no coupling between the
$\ket -$ subspace and the $\ket +$ subspace, the adiabatic theorem may
be applied within the $\ket -$ subspace.  The relevant eigenvalue gap
is thus the same as that of $H_{XZ}(t)$, and so is the runtime.

In standard adiabatic quantum computation, the qubits are in the ground
state of the instantaneous Hamiltonian. Thus, any disturbance to the
state costs energy. This is thought to offer some protection against
thermal noise~\cite{Childs}. When performing universal adiabatic
quantum computation with $\widetilde{H}_{XZ}(t)$, the qubits are
not in the ground state. Thus, it is possible for the system to thermally
relax out of the computational state. However, this can only occur
by disturbing the ancilla qubit out of the state $\ket -$. By protecting
the ancilla qubit, one can to a large degree protect the entire computation.
Note that an energy penalty against the ancilla qubit leaving the
state $\ket -$ would be non-stoquastic. This is why the above construction
fails to prove QMA-completeness and universal adiabatic quantum computation
using the \emph{ground state} of a stoquastic Hamiltonian, as we
expect it must, based on the complexity-theoretic results
of~\cite{Bravyi1,Bravyi2,Liu,Bravyi_Terhal}.

\section{QMA-Complete problems for Markov Matrices}
\label{Markov_section}

The second main result of our paper provides an example of a
QMA-complete classical problem: finding the lowest eigenvalue of a
symmetric Markov matrix. A matrix with all nonnegative entries, such
that the entries in any given column sum to one is called a
stochastic or Markov matrix.  These matrices are named after Markov
chains, which are stochastic processes such that given the present
state, the future state is independent of the past states. Suppose a
system has $d$ possible states. Then, its probability distribution at
time $t$ is described by the $d$-dimensional vector $x_{t}$ whose
entries are nonnegative and sum to one. If the system is evolving
according to a Markov process then its dynamics are completely
specified by the equation $x_{t+1}=Mx_{t}$ where $M$ is a $d\times d$
stochastic matrix. Note that, like quantum Hamiltonians, Markov
matrices often have tensor product structure. For example, suppose we
have two independent simultaneous Markov chains governed by
$x_{t+1}=Mx_{t}$ and $y_{t+1}=Ny_{t}$. Then their joint probability
distribution $z$ is governed by $z_{t+1}=(M\otimes N)z_{t}$.

Markov processes for which the Markov matrix is symmetric correspond
to random walks on undirected weighted graphs. (Self-loops are allowed
and correspond to diagonal matrix elements.) These matrices are doubly-stochastic:
the sum of the entries in any row or column is one. By the Perron-Frobenius
theorem, the highest eigenvalue of a symmetric stochastic matrix is
one, and the corresponding eigenvector is the uniform distribution.
The eigenvalue with next largest magnitude controls the rate of convergence
of the process to its fixed point. A symmetric stochastic matrix is
Hermitian and therefore one can also think of these matrices as Hamiltonians.

To prove that finding the lowest eigenvalue of a 3-local symmetric
stochastic matrix is QMA-complete, we again use a reduction from the
QMA-complete $H_{XZ}$ Hamiltonian of~\cite{Biamonte}. We must take
the opposite sign convention from equation \ref{alpha}: 
\begin{equation}
H_{XZ}=\sum_{k}\alpha_{k}S_{k},
\end{equation}
where the coefficients $\alpha_{k}$ are the same as before (all
positive) and $S_{k}=-T_{k}$. Now define: 
\begin{equation}
\hat{H}_{XZ}=\frac{1}{N} \sum_k \alpha_k \widetilde{S}_k
\label{eq:normalize}
\end{equation}
where 
\[
N = \sum_k \alpha_k
\]
and $\widetilde{S}_{k}$ is the permutation matrix obtained by
applying the replacement rules (\ref{replacement}) to $S_{k}$. By
construction, $\hat{H}_{XZ}$ is a 3-local, symmetric, doubly stochastic
matrix. We can rewrite $\hat{H}_{XZ}$ as 
\begin{equation}
\hat{H}_{XZ}=\frac{1}{N} \left( H_{XZ}\otimes \ket{-} \bra{-} +
\bar{H}_{XZ} \otimes \ket{+} \bra{+} \right),
\end{equation}
where $\bar{H}_{XZ} = \sum_k \alpha_k |S_k|$. Thus, to determine an
eigenvalue of $H_{XZ}$ to within $\pm\epsilon$ we must find the
corresponding eigenvalue of $\hat{H}_{XZ}$ to within
$\pm\epsilon/N$. Because $H_{XZ}$ is a two-local Hamiltonian on $n$
qubits with coupling strengths of order unity, $N$ is at most
$O(n^{2})$. Thus the problem of determining the eigenvalue of
$\hat{H}_{XZ}$ corresponding to the ground state of $H_{XZ}$ to
polynomial precision is QMA-hard.

To obtain a cleaner QMA-hard problem we would like to construct a
stochastic matrix whose \emph{lowest} eigenvalue is QMA-hard to
find. To do this, let 
\[
H_{p}=(1-p)\sigma^{+}_{n+1}+p\hat{H}_{XZ}.
\]
Here $\sigma^{+}_{n+1} = \ket{+}\bra{+} = \frac{1}{2}(\id+X_{n+1})$ acts
on the ancilla qubit, thereby giving it an energy penalty of size
$(1-p)$ against leaving 
the state $\ket -$. For $0\leq p\leq1$, $H_{p}$ is a stochastic
Hamiltonian. For $p<1/3$ the energy penalty is large enough that
the highest eigenvalue in the $\Ket -$ subspace lies below the lowest
eigenvalue in the $\Ket +$ subspace. In this case the lower half
of the spectrum of $H_{p}$ is the spectrum of $H_{XZ}$ scaled
by $p/N$, and the upper half of the spectrum
of $H_{p}$ is the spectrum of $\bar{H}_{XZ}$ scaled by $p/N$
and shifted up by $1-p$. 

Thus, we can obtain the ground energy of $H_{XZ}$ to polynomial precision
by computing the lowest eigenvalue of $H_{p}$ to a higher but
still polynomial precision. This reduction proves that finding the
lowest eigenvalue of $H_{p}$ to polynomial precision is QMA-hard.
Using the quantum algorithm for phase estimation, one easily shows
that the problem of estimating the lowest eigenvalue of $H_{p}$ is
contained in QMA (see~\cite{Kitaev_book}). Thus this problem is
QMA-complete.

\section{Frustration Free Adiabatic Computation}

It was stated in \cite{Bravyi_Terhal} that universal adiabatic quantum
computation can be performed in the ground state of a 5-local
frustration-free Hamiltonian. Let $U = U_L \ldots U_2 U_1$ be a
quantum circuit acting on $n$ qubits with $L = \mathrm{poly}(n)$
gates. Let 
\begin{equation}
\label{psij}
\ket{\psi_j} = U_j \ldots U_1 \ket{0}^{\otimes n}
\end{equation}
be a state of $n$ qubits corresponding to the $j\th$ state of the time evolution of a quantum circuit specified by gates $U_j$ and 
\begin{equation}
\label{unary}
\ket{c_t} = \ket{1^{t+1} 0^{L-t}}.
\end{equation}
be a state of $L+1$ clock qubits. Bravyi and Terhal construct a parametrized 5-local Hamiltonian $H(s)$
such that the ground state $\ket{\psi(s)}$ satisfies 
\begin{eqnarray*}
\ket{\psi(1)} & = & \frac{1}{\sqrt{L+1}} \sum_{j=0}^L
\ket{\psi_j} \ket{c_j} \\
\ket{\psi(0)} & = & \frac{1}{\sqrt{L+1}} \sum_{j=0}^L \ket{0^n} \ket{c_j}.
\end{eqnarray*}
We can think of first register in $\ket{\psi(s)}$ as consisting of
``work'' qubits on which the computation happens and the second
register in $\ket{\psi(s)}$ as being a clock containing a time written
in unary.

For $s \in [0,1]$, the minimal eigenvalue gap between the ground state
and first excited state of $H(s)$ is $O(1/L^2)$. By the adiabatic
theorem~\footnote{Many versions of the adiabatic therem have been
  proven. For one example see appendix F of \cite{mythesis}.},
$1/\mathrm{poly}(L)$ eigenvalue gap ensures that given $\ket{\psi_0}$,
one obtains $\ket{\psi_1}$ by applying $H(s)$ and varying $s$ from
zero to one over $\mathrm{poly}(L)$ time. By measuring the clock
register of $\ket{\psi(1)}$, one obtains the result $\ket{1^{L+1}}$
with probability $1/(L+1)$. If this result is obtained, one finds the
output of the circuit $U$ by measuring the first register of qubits in
the computational basis. By repeating this process with $O(L)$ copies
of $\ket{\psi(1)}$ one succeeds with high probability. Alternatively,
one can pad the underlying circuit with $L$ identity gates, in which
case each trial succeeds with probability $1/2$.

The construction from \cite{Bravyi_Terhal} invokes the fact that the
spectrum of $H(1)$ is independent of the form of the gates $U_j$. By
choosing a gate set which is composed of elements of simply connected
unitary groups such as $SU(2)$ and $SU(4)$ one may construct a
continuous path connecting each gate to the identity, and use a single
parameter $s$ to transform all gates from the identity to the final
circuit at once. The Hamiltonian at $s=0$ corresponds to the identity
circuit, and its ground state is the uniform superposition of the
clock states tensored with the initial data on the work qubits. In
this ground state, the qubits of the clock register are entangled. It is
standard to design adiabatic computations such that the initial
Hamiltonian has a product state as its ground state, because such
states should be easily produced by cooling or single-qubit
measurements. In this section we construct a modified version of the
construction from \cite{Bravyi_Terhal} that satisfies this condition
and is still frustration free.

Let $c(j)$ indicate the $j\th$ clock qubit and let $w(j)$ indicate the
$j\th$ work qubit. Let
\begin{eqnarray*}
H^{\mathrm{init}}_j & = & \ket{1}\bra{1}_{w(j)} \otimes \ket{10}
\bra{10}_{c(1),c(2)} \\
H^{\mathrm{clock}}_j & = & \ket{01}\bra{01}_{c(j-1),c(j)} \\
\end{eqnarray*}
For $j\in\{1,...,L-1\}$ define 
\begin{eqnarray*}
H^{\mathrm{prop}}_j(s) & = & s
  \ket{100}\bra{100}_{c(j),c(j+1),c(j+2)} + \\
& & (1-s) \ket{110}\bra{110}_{c(j),c(j+1),c(j+2)} - \\
& & \sqrt{s (1-s)}  \left(
U_j \otimes \ket{110}\bra{100}_{c(j),c(j+1),c(j+2)} + \right. \\
& & \left. U_j^\dag \otimes \ket{100}\bra{110}_{c(j),c(j+1),c(j+2)}
\right). \\
\end{eqnarray*}
and let 
\begin{eqnarray*}
H^{\mathrm{prop}}_L(s) & = & s
  \ket{10}\bra{10}_{c(L),c(L+1)} + \\
& & (1-s) \ket{11}\bra{11}_{c(L),c(L+1)} - \\
& & \sqrt{s (1-s)}  \left(
U_L \otimes \ket{11}\bra{10}_{c(L),c(L+1))} + \right. \\
& & \left. U_L^\dag \otimes \ket{10}\bra{11}_{c(L),c(L+1)}
\right). \\
\end{eqnarray*}
It can be directly verified that each  $H^{\mathrm{clock}}_j$, $H^{\mathrm{init}}_j$, and
$H^{\mathrm{prop}}_j(s)$ is a projector. Here, for convenience, we
define $s$ so that it varies from zero to one half rather than from
zero to one as is done in \cite{Bravyi_Terhal}. Our frustration-free
Hamiltonian is the following sum of projectors:
\begin{eqnarray*}
H^{\mathrm{clock}} & = & \ket{0}\bra{0}_{c(0)} + \sum_{j=1}^L
H^{\mathrm{clock}}_j \\
H^{\mathrm{init}} & = & \sum_{j=1}^n H^{\mathrm{init}}_j \\
H^{\mathrm{prop}}(s) & = & \sum_{j=1}^L H^{\mathrm{prop}}_j(s) \\
H^{\mathrm{FF}}(s) & = & H^{\mathrm{clock}} + H^{\mathrm{init}} +
H^{\mathrm{prop}}(s)
\end{eqnarray*}

If $U_1 \ldots U_L$ are chosen from a universal set of two-qubit
gates then $H^{\mathrm{FF}}(s)$ is an efficient 5-local
frustration-free adiabatic quantum computer. To see how this
Hamiltonian achieves universal adiabatic computation, we examine the
various terms one by one. The ground state of
$H^{\mathrm{FF}}(s)$ is the simultaneous zero eigenspace of
$H^{\mathrm{clock}}$, $H^{\mathrm{init}}$, and
$H^{\mathrm{prop}}$. $H^{\mathrm{clock}}$ commutes with
$H^{\mathrm{prop}}(s) + H^{\mathrm{init}}$ and provides an energy penalty
of at least unit size if the clock register is not in one of the unary
states $\ket{c_t} = \ket{1^{t+1} 0^{L-t}}$. Thus the low lying
spectrum of $H^{\mathrm{FF}}(s)$ is strictly contained in the ground
space of $H^{\mathrm{clock}}$.

For any bit string $x \in \{0,1\}^n$ and integer $j \in
\{1,2,\ldots,L\}$, let
\[
\ket{\chi_x^j} = \left( U_j U_{j-1} \ldots U_1 \ket{x} \right) \otimes
\ket{c_j},
\]
where $\ket{c_j}$ is as defined in equation \ref{unary}. (We also define
$\ket{\chi_x^0} = \ket{x} \otimes \ket{c_0}$.) There are $2^n(L+1)$
such states and they form an orthonormal basis for the ground space of
$H^{\mathrm{clock}}$. In this basis, $H^{\mathrm{prop}}(s) +
H^{\mathrm{init}}$ takes the block-diagonal form
\[
H^{\mathrm{prop}}(s) + H^{\mathrm{init}} = \bigoplus_{x \in \{0,1\}^n} M_x
\]
where
\[
M_x = \left[ \begin{array}{cccccc}
s+|x|  & -b &    &    &    &    \\
-b & 1  & -b &    &    &    \\
   & -b & 1  & -b &    &    \\
   &    & \p & \p & \p &    \\
   &    &    & -b & 1  & -b \\
   &    &    &    & -b & 1-s
\end{array} \right]
\]
is an $L+1$ by $L+1$ matrix and $b = \sqrt{s(1-s)}$. Here $|x|$ denotes the Hamming weight of the
bit string $x$. The appearance of $|x|$ is the sole manifestation of
$H^{\mathrm{init}}$. The rest of the matrix elements all come from the
``hopping'' action of $H^{\mathrm{prop}}(s)$.

$M_{00\ldots0}$ has the unique ground state
\begin{equation}
\label{newhist}
N \sum_{j=0}^L r^j \ket{\psi_j} \ket{c_j}
\end{equation}
where $\ket{\psi_j}$ and $\ket{c_j}$ are as defined in equations
\ref{psij} and \ref{unary}, $r = \sqrt{\frac{s}{1-s}}$, and $N$ is a
normalization factor. This constitutes the ground state of
$H^{\mathrm{FF}}(s)$. The first excited state of $M_{00\ldots0}$ has
energy  $1-2 \sqrt{s (1-s)} \cos \left( \frac{\pi}{L+1}
\right)$. Because of the direct sum structure of $H^{\mathrm{FF}}(s)$,
we can apply the adiabatic theorem directly to $M_{00\ldots0}$. The
runtime of the adiabatic algorithm is thus determined by the gap
between the ground and first excited states of $M_{00\ldots0}$. This
takes its minumum at $s=1/2$, where it is equal to $1-\cos \left(
\frac{\pi}{(L+1)} \right) = O(1/L^2)$. For questions of fault
tolerance it is also useful to know the eigenvalue gap between the
ground and first excited states of the full Hamiltonian
$H^{\mathrm{FF}}(s)$. The first excited energy of $H^{\mathrm{FF}}(s)$
is equal to the ground energy of $M_{10\ldots0}$, which is $1-2
\sqrt{s(1-s)} \cos \left( \frac{\pi}{2(L+1)} \right)$. Thus the
minumum eigenvalue gap of $H^{\mathrm{FF}}(s)$ occurs at $s=1/2$ and
is equal to $1-\cos \left( \frac{\pi}{2 (L+1)} \right) = O(1/L^2)$.

By equation \ref{newhist}, the ground state of $H^{\mathrm{FF}}(0)$ is
$\ket{000\ldots} \otimes \ket{1000\ldots}$, and the ground state of
$H^{\mathrm{FF}}(1/2)$ is the same state  $\frac{1}{\sqrt{L+1}}
\sum_{j=0}^L \ket{\psi_j} \ket{c_j}$ produced by the scheme of
\cite{Bravyi_Terhal}.

\section{Stochastic Frustration Free Computation}

In \cite{Bravyi_Terhal}, Bravyi and Terhal showed that adiabatic
quantum computation in the ground state of a stoquastic
frustration free Hamiltonian can be efficiently simulated by a
classical computer. In this section we show that in contrast, one can
perform universal adiabatic quantum computation in the ground state of 
a \emph{stochastic} frustration free (StochFF) Hamiltonian
$H^{\mathrm{StochFF}}(s)$. (Alternatively, we can view this as computation
in the highest energy state of the stoquastic Hamiltonian
$-H^{StochFF}(s)$.) 

It has been shown that the two-qubit CNOT gate, together with
any one-qubit rotation whose square is not basis preserving, are
sufficient to perform universal quantum computation~\cite{Shi}. All
matrix elements in these gates are real numbers. If we choose
$U_1,\ldots,U_L$ from this gate set then $H^{\mathrm{FF}}(s)$ is a
5-local real frustration free Hamiltonian. Examining the construction
of section \ref{Markov_section} one sees that it can be applied to any
Hamiltonian with real matrix elements, and it increases the locality
by one. This construction also preserves frustration-freeness, as we will show in the next paragraph.
We can thus use this construction on $H^{FF} (s)$ to obtain a
6-local stochastic frustration free Hamiltonian whose ground state is
universal for adiabatic quantum computation

To show that the mapping of section \ref{Markov_section} preserves frustration-freeness, consider applying this mapping to a frustration free local Hamiltonian $H=\sum_{j=1}^{m}H_j$, where $H_j=\sum_{k}\alpha_{k}^{j} S_k^j$ (where each $S_k^j$ is, up to an overall sign, a tensor product of Pauli operators and each $\alpha_k^j$ is positive).
We obtain the Hamiltonian
\begin{eqnarray}
H_p&=&p\hat{H}+\left(1-p\right)\left(\frac{\id+X_{n+1}}{2}\right)\nonumber \\
&=&\sum_j \frac{N_j}{N} \left[p\hat{H_j}+\left(1-p\right)\left(\frac{\id+X_{n+1}}{2}\right)\right] \label{fr_free}
\end{eqnarray}
where $N_j=\sum_k {\alpha_k^j}$ and $N=\sum_j N_j$. When $p<\frac{1}{3}$, $H_p$ is stochastic and has a zero energy ground state with an eigenvalue gap which is $\frac{p}{N}$ times the gap of $H$. Furthermore we see from \eqref{fr_free} (and the fact that each $H_j$ is positive semidefinite) that $H_p$ is a sum of positive semidefinite operators. Hence the  Hamiltonian $H_p$ is frustration free.
\section{Generalizations}

The constructions of sections \ref{Markov_section} and
\ref{QMA_section} replace Hamiltonians with real matrix elements of
both signs by computationally equivalent Hamiltonians with real
positive matrix elements. In this section we show that this technique
can be generalized to directly replace Hamiltonians with complex
matrix elements by computationally equivalent Hamiltonians with only
real positive matrix elements. However, in the process we necessarily
introduce a two-fold degeneracy of the ground state. 

Let $H$ be an arbitrary $k$-local Hamiltonian. We may expand $H$ as
\begin{equation}
\label{pdecomp}
H = \sum_{j} \alpha_j O_j
\end{equation}
where each $O_j$ is a tensor product of $k$ or fewer Pauli
matrices and each $\alpha_j$ is positive. Each entry in each $O_j$ is
$\pm 1$ or $\pm i$. We can replace the group $\{1,i,-1,-i\}$ with its
left-regular representation
\begin{equation}
\label{freplace}
\begin{split}i & \mapsto F\\
-1 & \mapsto F^{2}\\
-i & \mapsto F^{3}\\
1 & \mapsto F^{4}\end{split}
\end{equation}
where
\begin{equation}
F=\begin{pmatrix}0 & 1 & 0 & 0\\
0 & 0 & 1 & 0\\
0 & 0 & 0 & 1\\
1 & 0 & 0 & 0\end{pmatrix} .
\end{equation}
The eigenvectors of $F$ are $\ket{v_0},\ket{v_1},\ket{v_2},\ket{v_3}$
where
\[
\ket{v_j} = \frac{1}{2}\sum_{l = 0}^3 i^{lj} \ket{l}. 
\]
The corresponding eigenvalues are
\begin{equation}
\label{eigenvalues}
F \ket{v_j} = i^j \ket{v_j}.
\end{equation}

Let $S_j$ and $A_j$ be the real and imaginary parts of $\alpha_j
O_j$. That is, $S_j$ and $A_j$ are the unique real symmetric and
anti-symmetric matrices such that
\[
\alpha_j O_j = S + iA
\]
Further let $S_j^+ = (|S_j|+S_j)/2)$ and $S_j^- = (|S_j|-S_j)/2$ and
similarly for $A_j^{\pm}$, where $|\cdot|$ denotes the entrywise
absolute value. Applying the replacement (\ref{freplace}) to $H$ and
dividing by $N = \sum_j \alpha_j$ yields the stochastic Hamiltonian
$\tilde{H}$ with the decomposition
\begin{eqnarray}
\label{tildeH}
\tilde{H} & = \frac{1}{N} & \left( H^{(0)} \otimes \ket{v_0} \bra{v_0} +
H^{(1)} \otimes \ket{v_1} \bra{v_1} + \right. \nonumber \\
 & & \left. H^{(2)} \otimes \ket{v_2} \bra{v_2} +
     H^{(3)} \otimes \ket{v_3} \bra{v_3} \right)
\end{eqnarray}
where
\begin{eqnarray*}
H^{(0)} & = & \sum_j S_j^+ + S_j^- + A_j^+ + A_j^- \\
H^{(1)} & = & \sum_j S_j^+ - S_j^- + iA_j^+ - iA_j^- \\
H^{(2)} & =& \sum_j S_j^+ + S_j^- - A_j^+ - A_j^-\\
H^{(3)} & =  & \sum_j S_j^+ - S_j^- - iA_j^+ + iA_j^-. 
\end{eqnarray*}

$H^{(1)} = H$, thus the spectrum of $\tilde{H}$ in the $\ket{v_1}$
subspace matches that of $H$ up to a normalization factor of $N$ and a
pair of extra ancilla qubits. If we write each projector
$\ket{v_j}\bra{v_j}$ in terms of the Pauli basis we obtain
\begin{eqnarray*}
\ket{v_0}\bra{v_0} & = & \Pi_+^X \Pi_+^X \\
\ket{v_1}\bra{v_1} & = & \Pi_-^X \Pi_+^Y \\
\ket{v_2}\bra{v_2} & = & \Pi_+^X \Pi_-^X \\
\ket{v_3}\bra{v_3} & = & \Pi_-^X \Pi_-^Y \\
\end{eqnarray*}
where $\Pi_\pm^a$ is the projector onto the eigenvalue $\pm 1$ eigenstate of the Pauli matrix $a$.
Thus an $X$ penalty on the first ancilla qubit will separate the
$\ket{v_1},\ket{v_3}$ subspace from the $\ket{v_0},\ket{v_2}$
subspace. So, taking, $0<p < \frac{1}{3}$, the stochastic Hamiltonian 
\[
H_p^\prime = \left(1-p\right)\left(\frac{\id+X_{n+1}}{2}\right) + p\tilde{H}
\]
has ground space spanned by $\ket{\psi^{(1)}} \ket{v_1}$ and
$\ket{\psi^{(3)}} \ket{v_3}$, where $\ket{\psi^{(1)}}$ is the ground
state of $H^{(1)}$ and $\ket{\psi^{(3)}}$ is the ground state of
$H^{(3)}$. $H^{(1)} = H$, thus $\ket{\psi^{(1)}}$ is the ground state
of $H$. $H^{(3)} = H^*$, thus $\ket{\psi^{(3)}}$ is the complex
conjugate of the ground state of $H$.

A simple argument shows that the doubling in the spectrum of $H_p^\prime$ is a necessary property for
any construction which  maps an arbitrary Hamiltonian onto a real
Hamiltonian $H_R$, where $H_R$ is equal to $H$ within a fixed 1D
subspace of the ancillas.
Suppose that we have such a map which sends an arbitrary Hamiltonian
$H$ which acts on a Hilbert space $\mathcal{H}_1$ to a real
Hamiltonian $H_R$ on a  larger Hilbert space
$\mathcal{H}_1\otimes\mathcal{H}_2$ with the property that
\begin{equation}\label{HR}
H_R=H\otimes |\phi\rangle \langle \phi|+H^{\text{other}}\otimes(\id-|\phi\rangle\langle \phi|).
\end{equation}
where the state $|\phi\rangle\in\mathcal{H}_2$ does not depend on the
particular Hamiltonian $H$ but the operator $H^{\text{other}}$ may
depend on $H$. Then for any eigenvector $|\psi\rangle$ of $H$ with
energy $E$ we have:
\begin{equation}
H_R|\psi\rangle|\phi\rangle=E|\psi\rangle|\phi\rangle
\end{equation}
Since $H_R$ is real, complex conjugating this equation gives:
\begin{equation}\label{ccHR}
H_R|\psi^\star \rangle|\phi^\star \rangle=E|\psi^\star \rangle|\phi^\star \rangle
\end{equation}
To show that doubling exists in the spectrum it is sufficient to show
that $\langle \phi|\phi^\star\rangle=0$. To prove this, first use
equation \eqref{ccHR} to obtain
\begin{equation}
\left(\id \otimes|\phi\rangle\langle \phi|\right)H_R |\psi^\star \rangle |\phi^\star \rangle=E|\psi^\star \rangle  |\phi\rangle \langle \phi|\phi^\star\rangle 
\end{equation}
Then use equation \eqref{HR} to obtain
\begin{equation}
\left(\id \otimes|\phi\rangle\langle \phi|\right)H_R |\psi^\star \rangle |\phi^\star \rangle=\left(H|\psi^\star \rangle \right) |\phi\rangle \langle \phi|\phi^\star\rangle. 
\end{equation}
Equating these expressions gives
\begin{equation}
H|\psi^\star\rangle \langle \phi|\phi^\star\rangle=E|\psi^\star \rangle \langle \phi|\phi^\star\rangle. 
\end{equation}
This must hold for all Hamiltonians $H$ and eigenstates $|\psi\rangle$
and therefore it must be the case that $\langle
\phi|\phi^\star\rangle=0$. So we have shown that the doubling in the
spectrum of $H_p^\prime$ is a necessary feature of the type of maps we
consider. For constructing universal adiabatic quantum computers the
degeneracy induced by this construction may be problematic. However,
for proving complexity-theoretic completeness results it is often
irrelevant, as we see in the next section.

\section{Stochastic $k$-SAT}

The methods of the previous section can be used to show, roughly
speaking, that deciding whether or not a Hamiltonian which is a sum of positive semidefinite stochastic operators
is frustration free is as difficult as the general problem of deciding
whether a Hamiltonian is frustration free. In this section we
formalize this by defining a problem called stochastic $k$-SAT, which we
show to be $\mathrm{QMA}_1$-complete for $k=6$.

We first recall the definition of stoquastic $k$-SAT which is given in
~\cite{Bravyi_Terhal}. \\
\\
\textbf{Problem:} Stoquastic $k$-SAT \\
 \textbf{Input:} Input: A set of $k$-local Hermitian operators $\{H_{j}\}$
for $j\in\{1,\dots,m\}$ where $m=\mathrm{poly}(n)$ and a parameter $\epsilon>1/\mathrm{poly}(n)$ 
\begin{enumerate}
\item {Each $H_{j}$ is positive semidefinite} 
\item {Each $H_{j}$ has norm which is bounded by a polynomial in $n$.} 
\item {Every $H_{j}$ is stoquastic.} 
\end{enumerate}
\textbf{Output:} If $H=\sum_j H_j$ has a zero energy ground state, then this is
a YES instance. Otherwise if every eigenstate of $H$ has energy $>\epsilon$
then it is a NO instance.\\
\textbf{Promise:} Either the ground state
of H has energy 0, or else it has energy $>\epsilon$.\\
\\
The stoquastic $k$-SAT problem is therefore the problem of deciding if a
given stoquastic Hamiltonian that is a sum of positive definite
operators is frustration free, given that either this is the case or
else its ground energy exceeds $\epsilon$ ~\cite{Bravyi_Terhal}. Note that
this definition of stoquastic $k$-SAT looks somewhat different from the
definition of quantum $k$-SAT which was given in section \ref{sec_intro}, which
was stated entirely in terms of projectors. Given an instance of
stoquastic $k$-SAT, we can define operators $\Pi_j$ which project onto
the zero eigenspaces of the $H_j$. When the Hamiltonians $H_j$ are
stoquastic, these projectors are guaranteed to have nonnegative
matrix elements in the computational basis ~\cite{Bravyi_Terhal}.
So given an instance of stoquastic $k$-SAT with Hermitian positive
semidefinite operators $H_j$, it is possible to construct another
instance of stoquastic $k$-SAT with operators $\tilde{H}_j=\{1-\Pi_j\}$
that are all projectors.

We now define a problem called stochastic $k$-SAT, which is identical
to stoquastic $k$-SAT except that condition 3 is replaced by\\

$3'$. Every $H_j$ is a stochastic matrix.\\
\\
We note that there does not appear to be an equivalence between this
definition of stochastic $k$-SAT and the corresponding definition where
all the $H_j$ are (in addition) required to be projectors. 

Given these two definitions and the foregoing map from an arbitrary
Hamiltonian to a stochastic Hamiltonian, we now show how to reduce
any instance of quantum 4-SAT to an instance of stochastic 6-SAT. Starting with 
an instance of quantum 4-SAT specified by a set of projectors
$\{\Pi_j\}$ (for $j\in{1,...,m}$) we use  the map of the previous
section (with $p=\frac{1}{3}$ for concreteness) on each projector to
obtain a set of 6-local positive semidefinite stochastic Hamiltonians
$\{H_j\}$ where 
\begin{equation}
H_j= \frac{2}{3}\left(\frac{\id+X_{n+1}}{2}\right) + \frac{1}{3}\tilde{\Pi}_j
\end{equation}
(Note that $\tilde{\Pi}_j$ refers to the operator obtained by applying
the mapping from equation \eqref{tildeH}.)  If the 4-SAT instance is
satisfiable, then the stochastic 6-SAT instance will also be
satisfiable. Define $N_{max}$ to be the maximum value of $N$ obtained for one of the terms $\Pi_j$ when using the mapping of equation \eqref{tildeH}. If the 4-SAT instance is not satisfiable then for any state $|\phi\rangle$ there is some projector $\Pi_k$ such that
$\langle \phi|\Pi_k |\phi \rangle\geq \frac{\epsilon}{m}$. If we take
the parameter $\tilde{\epsilon}$ of the stochastic 6-SAT instance to be
related to the parameter $\epsilon$ of the quantum 4-SAT instance by
$\tilde{\epsilon}=\frac{\epsilon}{3m N_{max}}$ then the stochastic 6-SAT
instance will also be unsatisfiable.   Therefore stochastic 6-SAT is
$\mathrm{QMA}_1$ hard.  Stochastic 6-SAT is contained in $\mathrm{QMA}_1$ since every
instance of stochastic 6-SAT can be mapped to an instance of quantum
6-SAT by taking projectors $\Pi_j$ which project onto everything but
the zero eigenspaces of the $H_j$.

So $\mathrm{QMA}_{1}$ completeness of stochastic 6-SAT follows from  the results of Bravyi
~\cite{Bravyi4} on quantum $k$-SAT. This is in contrast to stoquastic
$k$-SAT, which is contained in MA for every constant
$k$~\cite{Bravyi_Terhal}.

\section{QMA-completeness for excited states}
\label{excited}

The local Hamiltonian problem refers specifically to ground state
energies. Similarly, we have formulated a computational problem based
on the highest energy of a given Hamiltonian. It is natural to ask
about the complexity of estimating the $c\th$ excited state. We can
formulate this as follows. Let $H$ be a $k$-local Hamiltonian on the
Hilbert space $\mathcal{H}$ of $n$ qubits. Let $\lambda_1 \leq
\lambda_2 \leq \ldots \leq \lambda_{2^n}$ denote the eigenvalues of
$H$, with corresponding eigenvectors $\ket{\psi_1}, \ket{\psi_2},
\ldots, \ket{\psi_{2^n}}$. The $(k,c,\epsilon)$-energy problem is as
follows.\\
\\
\textbf{Problem:} $(k,c,\epsilon)$-energy\\
\textbf{Input:} We are given a classical description of $H$, an integer
$c \geq 1$, and a pair of parameters $a,b$ such that $b-a = \epsilon >
1/\mathrm{poly}(n)$.\\
\textbf{Output:} If $\lambda_c \leq a$ answer YES. If $\lambda_c \geq b$
output NO. \\
\textbf{Promise:} $H$ is such that the answer is YES or NO. \\

In this section we will show that the $(k,c,\epsilon)$-energy problem is
QMA-complete for any $c = O(1)$. Showing QMA-hardness is the easier of
the two proofs. This can be achieved as follows. Let
$H^{(a)}$ and $H^{(b)}$ be a pair of $k$-local Hamiltonians on $n$
qubits, with spectra $\lambda_1^{(a)},\ldots,\lambda_{2^n}^{(a)}$,
$\ket{\psi_1^{(a)}},\ldots,\ket{\psi_{2^n}^{(a)}}$ and
  $\lambda_1^{(b)},\ldots,\lambda_{2^n}^{(b)}$,
  $\ket{\psi_1^{(b)}},\ldots,\ket{\psi_{2^n}^{(b)}}$, respectively. Then
\[
H^{(ab)} = H^{(a)} \otimes \ket{0} \bra{0} + H^{(b)} \otimes \ket{1} \bra{1}
\]
is a $(k+1)$-local Hamiltonian on $n+1$ qubits. Its complete set of
eigenvalues is
$\lambda_1^{(a)},\ldots,\lambda_{2^n}^{(b)};\lambda_1^{(b)},\ldots,\lambda_{2^n}^{(b)}$,
with corresponding eigenvectors
$\ket{\psi_1^{(a)}}\ket{0}, \ldots,
\ket{\psi_{2^n}^{(a)}}\ket{0}; \ket{\psi_1^{(b)}}\ket{1}, \ldots,
\ket{\psi_{2^n}^{(b)}}\ket{1}$. 
To prove QMA-hardness of a low-lying excited state let $H_0$
to be a Hamiltonian such that determining whether the ground energy is
close to zero is QMA-hard. Given an integer $c$, let $d = \lceil
\log_2 c \rceil$, $P_k = \frac{1}{2} (Z_k + \id )$, and
\[
H_c = \sum_{k=0}^d 2^k P_k + \sum_{k=d+1}^n
2^{d+1} P_k - \left( c-\frac{1}{2} \right) \id
\]
$H_c$ has exactly $c$ states with negative energy, and its
lowest nonnegative eigenvalue is $\frac{1}{2}$. Thus determining the
$c\th$ excited energy of $H_c \otimes \ket{0} \bra{0} + H_0 \otimes
\ket{1} \bra{1}$ is QMA-hard. In particular, it is interesting to
note that by choosing $c=2$ we construct a Hamiltonian whose
eigenvalue gap between the ground state and first excited state is
QMA-hard to compute.

Next we show containment in QMA. The
naive protocol would be for Merlin to provide Arthur with the state
$\ket{\psi_1} \ket{\psi_2} \ldots \ket{\psi_c}$ and for Arthur to use
phase estimation to check that the $c$ registers each contain a state
of energy at most $a$. The problem is that for ``no'' instances there
are many ways for Merlin to cheat. For example if $\lambda_1 \leq a$
but $\lambda_c \geq b$, the answer is ``no'' but Merlin can provide
the state $\ket{\psi_1}^{\otimes c}$ as a supposed witness. To prevent
this, Arthur needs to somehow check that he has been given a set of
$c$ orthogonal states that each have energy at most $a$. Thus we
propose the following protocol.

Arthur demands that Merlin give him the state
\begin{equation}
\label{witness}
\ket{W} = \frac{1}{\sqrt{c!}} \sum_{\pi \in S_c} \mathrm{sign}(\pi)
\ket{\psi_{\pi(1)}} \ket{\psi_{\pi(2)}} \ldots \ket{\psi_{\pi(c)}}.
\end{equation}
Arthur performs the projective measurement to see that the state given
to him by Merlin lies in the antisymmetric subspace of
$\mathcal{H}^{\otimes c}$. If this fails he rejects. He then throws
away all but the first register and performs phase estimation of $H$
to precision better than $\epsilon$. If the state has energy above $b$
he rejects. Otherwise he accepts.

It is clear that for YES instances, Arthur will accept the state
$\ket{W}$ with high probability. (The only source of error is
imprecision in phase estimation.) We will next prove that for NO
instances the acceptance probability is at most $1-\frac{1}{c}$. Using
standard methods\cite{Kitaev_book,MW,NWZ} we can amplify this protocol to
obtain polynomially small acceptance probability for NO instances.

\begin{lemma}
\label{symm}
For any state $\ket{\phi}$ in the antisymmetric subspace of
$\mathcal{H}^{\otimes c}$ and any state $\ket{\alpha_1} \in
\mathcal{H}$, $\bra{\phi} \left( \ket{\alpha_1} \bra{\alpha_1} \otimes
\id \right) \ket{\phi} \leq \frac{1}{c}$ where $\id$ is the identity
operator on $\mathcal{H}^{\otimes (c-1)}$.
\end{lemma}

\begin{proof}
Extend $\ket{\alpha_1}$ to an orthonormal basis
$\ket{\alpha_1},\ket{\alpha_2},\ldots,\ket{\alpha_{2^n}}$ for
$\mathcal{H}$. Let $F$ be the set of functions $f:\{1,2,\ldots,c\} \to
\{1,2,\ldots,2^n\}$ such that $f(1) < f(2) < \ldots < f(c)$. Thus $|F|
= \binom{2^n}{c}$. For any $f \in F$ we have the corresponding Slater
determinant state.
\[
\ket{D_f} = \frac{1}{\sqrt{c!}} \sum_{\pi \in S_c} \mathrm{sign}(\pi)
\ket{\alpha_{f(\pi(1))}} \ket{\alpha_{f(\pi(2))}} \ldots
\ket{\alpha_{f(\pi(c))}}
\]
It a standard result that these $\binom{2^n}{c}$ states form a
complete orthonormal basis for the antisymmetric subspace of
$\mathcal{H}^{\otimes c}$. For any $f,g \in F$ we have
\begin{equation}
\label{deltas}
\bra{D_f} \left( \ket{\alpha_1} \bra{\alpha_1} \otimes \id \right)
\ket{D_g} = \frac{ \delta_{f,g} \delta_{f(1),1}}{c}
\end{equation}
where each $\delta$ denotes a generalized Kronecker-$\delta$. Because
$\ket{\phi}$ is antisymmetric, it can be decomposed in the Slater
determinant basis:
\[
\ket{\phi} = \sum_{f \in F} \phi_f \ket{D_f}
\]
and $p_f = \phi^*_f \phi_f$ is a corresponding probability
distribution on $F$. Thus
\[
\bra{\phi} \left( \ket{\alpha_1} \bra{\alpha_1} \otimes \id \right)
\ket{\phi} =
\] 
\[
\sum_{f,g \in F} \phi^*_f \bra{D_f} \left( \ket{\alpha_1}
\bra{\alpha_1} \otimes \id \right) \ket{D_g} \phi_g.
\]
By equation \ref{deltas} this is
\[
= \frac{1}{c} \sum_{f \in F} p_f \delta_{f(1),1} \leq \frac{1}{c}.
\]
\end{proof}

The quantity
\[
p_j(\phi) = \bra{\phi} \left( \ket{\psi_j}\bra{\psi_j} \otimes \id \right)
\ket{\phi}
\]
is the probability of obtaining $\ket{\psi_j}$ if we measure the first
register of a state $\ket{\phi}$ in the eigenbasis of $H$. By lemma
\ref{symm}
\[
\sum_{j=1}^{c-1} p_j(\phi) \leq 1 - \frac{1}{c}.
\]
Thus, with probability at least $\frac{1}{c}$, such a measurement
would yield $\ket{\psi_j}$ with $j > c-1$. Thus if $\lambda_c > b$
then with probability at least $\frac{1}{c}$ a measurement 
of the observable $H$ would yield energy at least $b$. The phase
estimation algorithm can in $\mathrm{poly}(1/\epsilon)$ time perform
such an energy measurement with exponentially small chance of making
an error as large as $\epsilon$. Thus the protocol is sound, which
completes the proof that the $(c,k,\epsilon)$-energy problem is
QMA-complete for constant $c$ and $k$ and polynomially small
$\epsilon$.

\section{Discussion and conclusions}

The results presented in this paper have several applications. Although
calculating the ground state energy of stoquastic Hamiltonians appears
easier than calculating the ground energy of generic Hamiltonians,
our results suggest that calculating other eigenstates of stoquastic
Hamiltonians remains hard. Because the wavefunctions of these states
have amplitudes which are both positive and negative, the hardness
of determining their energy supports the intuition that it is the
positivity of the amplitudes which makes the ground state problem
for stoquastic Hamiltonians easier. An extreme distinction between stochastic 
and stoquastic Hamiltonians arises when the Hamiltonians are also frustration free.
 Although adiabatic evolution with stoquastic frustration free Hamiltonians
is simulable in  BPP~\cite{Bravyi_Terhal}, we have shown that adiabatic evolution in 
the ground state of a stochastic frustration free Hamiltonian is universal.

Secondly, these results may be relevant for the physical implementation
of quantum computers. The first proof of universality of adiabatic
quantum computation used 5-local interactions~\cite{Aharonov_adiabatic}.
Since then, the Hamiltonians have been brought into incrementally
more physically feasible form by various techniques while retaining
universality~\cite{Kempe,Kempe_Regev,Oliveira}. The universal Hamiltonian
$H_{XZ}$ of~\cite{Biamonte} is one outcome of this chain of reductions.
Here we add one more step to this chain, obtaining universal stochastic
and stoquastic Hamiltonians which resemble those arising in some systems
of superconducting qubits~\cite{Bravyi1}. The constructions given here
are at least three local, and so would require the use of perturbative
gadgets to implement in terms of physical two-local interactions.

Finally, our results are of interest from a purely
complexity-theoretic point of view. Stochastic matrices arise outside
the context of quantum mechanics, in Markov chains. Our reduction
shows that finding the lowest eigenvalue of a certain class of
exponentially large but efficiently describable doubly-stochastic
matrices is QMA-complete. (These stochastic matrices correspond to Markov
Chains in which the ``update rule'' is a probabilistic selection over
some set of updates which are local in the tensor product sense.) In
general, the problem of finding eigenvalues of stochastic matrices is
of interest because the eigenvalue of second largest magnitude
determines the mixing time of the corresponding Markov chain. There
exist Markov chains in which the eigenvalue of second largest
magnitude is negative, and is the lowest lying eigenvalue. We hope
that the demonstration of a QMA-complete problem arising in a
classical setting will help shed further light on the class QMA
itself.

We thank Sergey Bravyi for a helpful discussion. SJ thanks MIT's
Center for Theoretical Physics, RIKEN's Digital Materials Laboratory,
and Caltech's Institute for Quantum Information, Franco Nori, Sahel
Ashhab, ARO/DTO's QuaCGR program, the DOE, the Sherman Fairchild
Foundation, and the NSF under grant number PHY-0803371. DG
thanks Eddie Farhi, the W. M. Keck Foundation Center for
Extreme Quantum Information Theory, and the Natural Sciences and
Engineering Research Council of Canada. PJL thanks John Preskill and
the Institute of Quantum Information at Caltech for hosting an
extended visit during which part of this work was completed. This
research was supported in part by the National Science Foundation
under Grant No. PHY05-51164.

\bibliography{stoquastic}

\end{document}